\pdfoutput=1
\documentclass[runningheads]{llncs}
\usepackage{amssymb}
\usepackage{amsmath}
\usepackage{graphicx}
\usepackage{verbatim}
\usepackage{color}
\usepackage{ae}
\usepackage[ruled]{algorithm}
\usepackage{algorithmicx}
\usepackage{algpseudocode}
\usepackage[breaklinks,bookmarks=false]{hyperref}

\usepackage{thmtools,thm-restate}

\DeclareMathOperator{\bord}{border}

\DeclareMathOperator{\prefix}{prefix}
\DeclareMathOperator{\suffix}{suffix}
\DeclareMathOperator{\nr}{nr}
\DeclareMathOperator{\rank}{rank}

\newcommand{\twodots}{\mathinner{\ldotp\ldotp}}

\newcommand{\proc}[1]{\textnormal{\scshape#1}}

\begin{document}

\title{Pattern matching in Lempel-Ziv compressed strings: fast, simple, and deterministic\thanks{Supported by MNiSW grant number N~N206 492638, 2010--2012}}
\titlerunning{Pattern matching in Lempel-Ziv compressed strings}

\author{Pawe\l{} Gawrychowski}
\institute{Institute of Computer Science,\\
	University of Wroc{\l}aw,\\
	ul. Joliot-Curie 15, 50--383 Wroclaw,
	Poland \\
	\email{gawry@cs.uni.wroc.pl}
	}

\maketitle
\begin{abstract}
Countless variants of the Lempel-Ziv compression are widely used in many real-life applications. This paper is concerned with a natural modification of the classical pattern matching problem inspired by the popularity of such compression methods: given an uncompressed pattern $s[1\twodots m]$ and a Lempel-Ziv representation of a string $t[1\twodots N]$, does $s$ occur in $t$? Farach and Thorup~\cite{Farach} gave a randomized $\mathcal{O}(n\log^2\frac{N}{n}+m)$ time solution for this problem, where $n$ is the size of the compressed representation of $t$. Building on the methods of~\cite{CharikarApproximation} and~\cite{GawrychowskiLZW}, we improve their result by developing a faster and fully deterministic $\mathcal{O}(n\log\frac{N}{n}+m)$ time algorithm with the same space complexity. Note that for highly compressible texts, $\log\frac{N}{n}$ might be of order $n$, so for such inputs the improvement is very significant. A (tiny) fragment of our method can be used to give an asymptotically optimal solution for the substring hashing problem considered by Farach and Muthukrishnan~\cite{FarachHashing}.

\end{abstract}

\section{Introduction}

Effective compression methods allow us to decrease the space requirements which is clearly worth pursuing on its own. On the other hand, we do not want to store the data just for the sake of having it: we want to process it efficiently on demand. This suggest an interesting direction: can we process the data without actually decompressing it? Or, in other words, can we speed up processing if the compression ratio is high? Answer to such questions clearly depends on the particular compression and processing method chosen. In this paper we focus on Lempel-Ziv (also known as LZ77, or simply LZ for the sake of brevity), one of the most commonly used compression methods being the basis of the widely popular \texttt{zip} and \texttt{gz} archive file formats, and on pattern matching, one of the most natural text processing problem we might encounter. More specifically, we deal with the compressed pattern matching problem: given an uncompressed pattern $s[1\twodots m]$ and a LZ representation of a string $t[1\twodots N]$, does $s$ occur in $t$? This line of research has been addressed before quite a few times already. Amir, Benson, and Farach~\cite{Amir} considered the problem with LZ replaced by Lempel-Ziv-Welch (a simpler and easier to implement specialization of LZ), giving two solutions with complexities $\mathcal{O}(n\log m+m)$ and $\mathcal{O}(n+m^2)$, where $n$ is the size of the compressed representation. The latter has been soon improved~\cite{Kosaraju} to $\mathcal{O}(n+m^{1+\epsilon})$. Then Farach and Thorup~\cite{Farach} considered the problem in its full generality and gave a (randomized) $\mathcal{O}(n\log^2\frac{N}{n}+m)$ time algorithm for the LZ case. Their solution consists of two phases, called {\it winding} and {\it unwinding}, the first one uses a cleverly chosen potential function, and the second one adds fingerprinting in the spirit of string hashing of Karp and Rabin~\cite{KarpRabin}. While a recent result of~\cite{Iacono} shows that the winding can be performed in just $\mathcal{O}(n\log\frac{N}{n})$, it is not clear how to use it to improve the whole running time (or remove randomization).  In this paper we take a completely different approach, and manage to develop a $\mathcal{O}(n\log\frac{N}{n}+m)$ time algorithm. This complements our recent result from SODA'11~\cite{GawrychowskiLZW} showing that in case of Lempel-Ziv-Welch, the compressed pattern matching can be solved in optimal linear time. The space usage of the improved algorithm is the same as in the solution of Farach and Thorup, $\mathcal{O}(n\log\frac{N}{n}+m)$.

Besides the algorithm of Farach and Throup, the only other result that can be applied to the LZ case we are aware of is the work of Kida \emph{et al.}~\cite{Kida}. They considered the so-called \emph{collage systems} allowing to capture many existing compression schemes, and developed an efficient pattern matching algorithm for them. While it does not apply directly to the LZ compression, we can transform a LZ parse into a non-truncating collage system with a slight increase in the size, see section~\ref{section:constructing}. The running time (and space usage) of the resulting algorithm is $\mathcal{O}(n\log\frac{N}{n}+m^{2})$. While $m^{2}$ might be acceptable from a practical point of view, removing the quadratic dependency on the pattern length seems to be a nontrivial and fascinating challenge from a more theoretical angle, especially given that for some highly compressible texts $n$ might be much smaller than $m$. Citing~\cite{Kida}, even decreasing the dependency to $m^{1.5}\log m$ (the best preprocessing complexity known for the LZW case~\cite{Kosaraju} at the time) ``is a challenging problem''.

While we were not able to achieve linear time for the general LZ case, the algorithm developed in this paper not only significantly improves the previously known time bounds, but also is fully deterministic and (relatively) simple. Moreover, LZ compression allows for an exponential decrease in the size of the compressed text, while in LZW $n$ is at least $\sqrt{N}$. In order to deal with such highly compressible texts efficiently we need to combine quite a few different ideas, and the nonlinear time of our (and the previously known) solution might be viewed as an evidence that LZ is substantially more difficult to deal with than LZW. While most of those ideas are simple, they are very carefully chosen and composed in order to guarantee the $\mathcal{O}(n\log\frac{N}{n}+m)$ running time. We believe the simplicity of those basic building blocks should not be viewed as a drawback. On the contrary, it seems to us that improving a previously known result (which used fairly complicated techniques) by a careful combination of simple tools should be seen as an advantage. We also argue that in a certain sense, our result is the best possible: if integer division is not allowed, our algorithm can be implemented in $\mathcal{O}(n\log N+m)$ time, and this is the best time possible.

\section{Overview of the algorithm}

Our goal is to detect an occurrence of $s$ in a given Lempel-Ziv compressed text $t[1\twodots N]$. The Lempel-Ziv representation is quite difficult to work with efficiently, even for a such simple task as extracting a single letter. The starting point of our algorithm is thus transforming the input into a {\it straight-line program}, which is a context-free grammar with each nonterminal generating exactly one string. For that we use the method of Charikar {\it et al.}~\cite{CharikarApproximation} to construct a SLP of size $\mathcal{O}(n\log\frac{N}{n})$ with additional property that all productions are {\it balanced}, meaning that the right sides are of the form $XY$ with $\frac{\alpha}{1-\alpha}\leq\frac{|X|}{|Y|}\leq\frac{1-\alpha}{\alpha}$ for some constant $\alpha$, where $|X|$ is the length of the (unique) string generated by $X$. Note that Rytter gave a much simpler algorithm~\cite{RytterApproximation} with the same size guarantee, using the so-called AVL grammars but we need the grammar to be balanced. We also need to add a small modification to allow self-referential LZ.

After transforming the text into a balanced SLP, for each nonterminal we try to check if the string it represents occurs inside $s$, and if so, compute the position of (any) its occurrence. Otherwise we would like to compute the longest prefix (suffix) of this string which is a suffix (prefix) of $s$. At first glance this might seem like a different problem that the one promised to solve: instead of locating an occurrence of the pattern in the text, we retrieve the positions of fragments of the text in the pattern. Nevertheless, solving it efficiently gives us enough information to answer the original question due to a constant time procedure which detects an occurrence of $s$ in a concatenation of two its substrings. 

The first (simple) algorithm for processing a balanced SLP we develop requires as much as $\mathcal{O}(\log m)$ time per query, which results in $\mathcal{O}(n\log\frac{N}{n}\log m+m)$ total complexity. This is clearly not enough to beat~\cite{Farach} on all possible inputs. Hence instead of performing the computation for each nonterminal separately, we try to process them in $\mathcal{O}(\log N)$ groups corresponding to the (truncated) logarithm of their length. Using the fact that the grammar is balanced, we are then able to achieve $\mathcal{O}(n\log\frac{N}{n}+m\log m)$ time.  Because of some technical difficulties, in order to decrease this complexity we cannot really afford to check if the represented string occurs in $s$ for each nonterminal exactly, though. Nevertheless, we can compute some approximation of this information, and by using a tailored variant of binary search applied to all nonterminals in a single group at once, we manage to process the whole grammar in time proportional to its size while adding just $\mathcal{O}(m)$ to the running time. 

%Because of the space constraints, some proofs are in the appendix.

\section{Preliminaries}

The computational model we are going to use is the standard RAM allowing direct and indirect addressing, addition, subtraction, integer division and conditional jump with word size $w\geq\max\{\log n,\log N\}$. One usually allows multiplication as well in this model but we do not need it, and the only place where we use integer division (which in some cases is known to significantly increase the computational power), is the proof of Lemma~\ref{lemma:balanced construction}.

We do not assume that any other operation (like, for example, taking logarithms) can be performed in constant time on arbitrary words of size $w$. Nevertheless, because of the $n$ addend in the final running time, we can afford to preprocess the results on words of size $\log n$ and hence assume that some additional (reasonable) operations can be performed in constant time on such inputs.

As usually, $|w|$ stands for the length of $w$, $w[i\twodots j]$ refers to its fragment of length $j-i+1$ beginning at the $i$-th character, where characters are numbered starting from $1$. All strings are over an alphabet $\Sigma$ of polynomial cardinality, namely $\Sigma=\{1,2,\ldots,(n+m)^c\}$.  A border of $w[1\twodots |w|]$ is a fragment which is both a prefix and a suffix of $w$, i.e., $w[1\twodots i]=w[|w|-i+1\twodots |w|]$. We identify such fragment with its length and say that $\bord(t)=\{i_1,\ldots,i_k \}$ is the set of all borders of $t$. A period of a string $w[1\twodots |w|]$ is an integer $p$ such that $w[i]=w[i+p]$ for all $1\leq i\leq |w|-p$. Note that $p$ is a period of iff $|w|-p$ is a border. The following lemma is a well-known property of periods.

\begin{lemma}[Periodicity lemma]\label{lemma:periodicity}
If $p$ and $q$ are both periods of $w$, and $p+q\leq |w|+\gcd(p,q)$, then $\gcd(p,q)$ is a period as well.
\end{lemma}

The Lempel-Ziv representation of a string $t[1\twodots N]$ is a sequence of triples $(start_i,len_i,next_i)$ for $i=1,2,\ldots,n$, where $n$ is the size of the representation. $start_i$ and $len_i$ are nonnegative integers, and $next_i\in\Sigma$. Such triple refers to a fragment of the text $t[start_i\twodots start_i+len_i-1]$ and defines $t[1+\sum_{j<i}len_j\twodots\sum_{j\leq i}len_j]=t[start_i\twodots start_i+len_i-1]next_i$. We require that $start_i\leq\sum_{j<i}len_j$ if $len_i>0$. The representation is not self-referential if all fragments we are referring to are already defined, i.e., $start_i+len_i-1\leq\sum_{j<i}len_j$ for all $i$. The sequence of triples is often called the {\it LZ parse} of text.

{\it Straight-line program} is a context-free grammar in the Chomsky normal form such that the nonterminals $X_1,X_2,\ldots,X_s$ can be ordered in such a way that each $X_i$ occurs exactly once as a left side, and whenever $X_i\rightarrow X_j X_k$ it holds that $j,k<i$. We identify each nonterminal with the unique string it derives, so $|X|$ stands for the length of the string derived from $X$. We call a straight-line program (SLP) {\it balanced} if for each production $X\rightarrow YZ$ both $|Y|$ and $|Z|$ are bounded by a constant fraction of $|X|$.

We preprocess the pattern $s$ using standard tools (suffix trees~\cite{Ukkonen} built for $s$ and reversed $s$, and LCA queries~\cite{BenderLCA}) to get the following primitives.

\begin{lemma}\label{lemma:equality}
Pattern $s$ can be preprocessed in linear time so that given $i,j,k$ representing any two fragments $s[i\twodots i+k]$ and $s[j\twodots j+k]$ we can find their longest common prefix (suffix) in constant time.
\end{lemma}

\begin{restatable}{lemma}{lemmalongestsuffix}
\label{lemma:longest suffix}
Pattern $s$ can be preprocessed in linear time so that given any fragment $s[i\twodots j]$ we can find its longest suffix (prefix) which is a prefix (suffix) of the whole pattern in constant time, assuming we know the (explicit or implicit) vertex corresponding to $s[i\twodots j]$ in the suffix tree built for $s$ (reversed $s$).
\end{restatable}

\begin{proof}
We assume that the suffix tree is built for $s$ concatenated with a special terminating character, say $\$$. Each leaf in the suffix tree corresponds to some suffix of $s$, and is connected to its parent with an edge labeled with a single letter. If we mark all those parents, finding the longest prefix which is a suffix of the whole $s$ reduces to finding the lowest marked vertex on a given path leading the root, which can be precomputed for all vertices in linear time.
\qed
\end{proof}

We will also use the suffix array $SA$ built for $s$~\cite{Karkainnen}. For each suffix of $s$ we store its position inside $SA$, and treat the array as a sequence of strings rather than a permutation of $\{1,2,\ldots,|s|\}$. Given any word $w$, we will say that it occurs at position $i$ in the $SA$ if $w$ begins $s[SA[i]\twodots |s|]$. Similarly, the fragment of $SA$ corresponding to $w$ is the (maximal) range of entries at which $w$ occurs.

\section{Snippets toolbox}

In this section we develop a few efficient procedures operating on fragments of the pattern, which we call {\it snippets}:

\begin{definition}
A snippet is a substring of the pattern $s[i\twodots j]$. If $i=1$ we call it a prefix snippet, if $j=m$ a suffix snippet.
\end{definition}

We identify snippets with the substrings they represent, and use $|s|$ to denote the length of the string represented by $s$. A snippet is stored as a pair $(i,j)$.

The two results of this section that we are going to use later build heavily on the contents of~\cite{GawrychowskiLZW}. Specifically, Lemma~\ref{lemma:concatenation occurrence} appears there as Lemma 5. To prove it, we first need the following simple and relatively well known property of borders.

\begin{lemma}\label{lemma:few borders}
If the longest border of $t$ is of length $b\geq\frac{|t|}{2}$ then all borders of length at least $\frac{|t|}{2}$ create one arithmetic progression. More specifically, $\bord(t)\cap\left\{\frac{|t|}{2},\ldots,|t|\right\}=\left\{|t|-\alpha p: 0\leq\alpha\leq\frac{|t|}{2p} \right\}$, where $p=|t|-b$ is the period of $t$. We call this set the long borders of $t$.
\end{lemma}

By applying the preprocessing from the Knuth-Morris-Pratt algorithm to $s$ and $s^r$ we can extract borders of prefix and suffix snippets efficiently.

\begin{lemma}\label{lemma:borders preprocessing}
Pattern $s$ can be preprocessed in linear time so that we can find the longest border of each its prefix (suffix) in constant time.
\end{lemma}

%\begin{proof}
%Apply the preprocessing from the Knuth-Morris-Pratt algorithm to both $s$ and reversed $s$. This results in computing the borders for all prefixes and %suffixes.
%\qed
%\end{proof}

The first result tells how to detect an occurrence in a concatenation of two snippets. We will perform a lot of such operations.
%, and an efficient implementation is crucial.

\begin{restatable}[Lemma 5 of~\cite{GawrychowskiLZW}]{lemma}{lemmaconcatenationoccurrence}
\label{lemma:concatenation occurrence}
Given a prefix snippet and a suffix snippet we can detect an occurrence of the pattern in their concatenation in constant time.
\end{restatable}

\begin{proof}
We need to answer the following question: does $s$ occur in $s[1\twodots i]s[j\twodots m]$? Or, in other words, is there $x\in border(s[1\twodots i])$ and $y\in border(s[j\twodots m])$ such that $x+y=m$? Note that either $x\geq\frac{|s[1\twodots i]|}{2}$ or $y\geq\frac{|s[j\twodots m]|}{2}$, and without losing the generality assume the former. From Lemma~\ref{lemma:few borders} we know that all such possible values of $x$ create one arithmetic progression. More specifically, $x=i-\alpha p$, where $p\leq\frac{i}{2}$ is the period of $s[1\twodots i]$ extracted using Lemma~\ref{lemma:borders preprocessing}. We need to check if there is an occurrence of $s$ in $s[1\twodots i]s[j\twodots m]$ starting after the $\alpha p$-th character, for some $0\leq\alpha\leq\frac{i}{p}$. For any such possible interesting shift, there will be no mismatch in $s[1\twodots i]$. There might be a mismatch in $s[j\twodots m]$, though.

Let $k\geq i$ be the longest prefix of $s$ for which $p$ is a period (such $k$ can be calculated efficiently by looking up the longest common prefix of $s[p+1\twodots m]$ and the whole $s$). We shift $s[1\twodots k]$ by $\left\lfloor \frac{\min(i,i-j+1)}{p}\right\rfloor p$ characters. Note this is the maximal shift of the form $\alpha p$ which, after extending $s[1\twodots k]$ to the whole $s$, does not result in sticking out of the right end of $s[j\twodots m]$. Then compute the leftmost mismatch of the shifted $s[1\twodots k]$ with $s[j\twodots m]$, see Figure~\ref{figure:concatenation}. Position of the first mismatch, or its nonexistence, allows us to eliminate all but one interesting shift. More precisely, we have two cases to consider.

\begin{enumerate}

\item There is no mismatch. If $k=m$ we are done, otherwise $s[k+1]\neq s[k+1-p]$, meaning that choosing any smaller interesting shift results in a mismatch.

\item There is a mismatch. Let the conflicting characters be $a$ and $b$ and call the position at which $a$ occurs in the concatenation the {\it obstacle}. Observe that we must choose a shift $\alpha p$ so that $s[1\twodots k]$ shifted by $\alpha p$ is completely on the left of the obstacle. On the other hand, if $s[1\twodots k]$ shifted by $(\alpha + 1) p$ is completely on the left as well, shifting $s[1\twodots k]$ by $\alpha p$ results in a mismatch because $s[k+1]\neq s[k+1-p]$ and $s[k+1-p]$ matches with the corresponding character in $s[j\twodots m]$. Thus we may restrict our attention to the largest shift for which $s[1\twodots k]$ is on the left of the obstacle.

\end{enumerate}

Having identified the only interesting shift, we verify if there is a match using one longest common prefix query on $s$. More precisely, if the shift is $\alpha p$, we check if the common prefix of $s[i-\alpha p\twodots m]$ and $s[j\twodots m]$ is of length $|s[i-\alpha p\twodots m]|$. Overall, the whole procedure takes constant time.
\qed

\begin{figure}
\centering
\includegraphics[width=\linewidth]{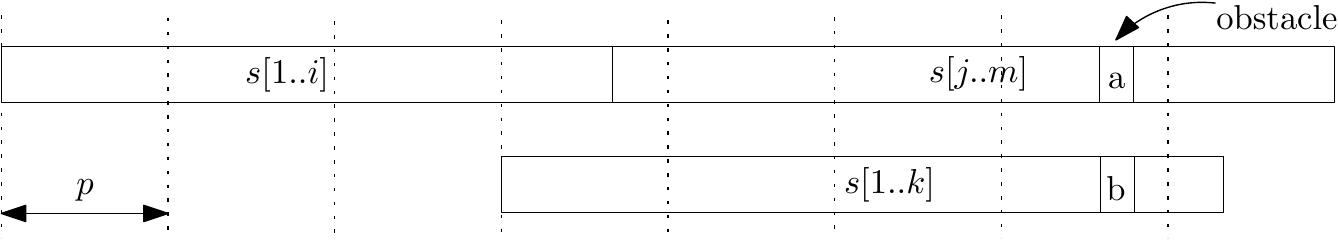}
\caption{Detecting an occurrence in a concatenation of two snippets.}
\label{figure:concatenation}
\end{figure}

\end{proof}

The second result can be deduced from Lemma 6 and Lemma 8 of~\cite{GawrychowskiLZW}, but we prefer to give an explicit proof for the sake of completeness. Its running time is constant as long as $|s_1|$ is bounded from above by a constant fraction of $|s_2|$.

\begin{restatable}{lemma}{lemmaconcatenationprefix}
\label{lemma:concatenation prefix}
Given a prefix snippet $s_1$ and a snippet $s_2$ for which we know the corresponding (explicit or implicit) node in the suffix tree, we can compute the longest prefix of $s$ which is a suffix of $s_1 s_2$ in time $\mathcal{O}\left(\max\left(1,\log\frac{|s_1|}{|s_2|}\right)\right)$.
\end{restatable}

\begin{proof}
We try to find the longest border of $s_1=s[1\twodots i]$ which can be extended with $s_2$. If there is none, we use Lemma~\ref{lemma:longest suffix} on $s_2$ to extract the answer. Of course $s_1$ might happen to have quite a lot of borders, and we do not have enough time to go through each of them separately. We try to abuse Lemma~\ref{lemma:few borders} instead: there are just $\log |s_1|$ groups of borders, and we are going to process each of them in constant time. It is not enough though, we need something faster when $|s_2|$ is relatively big compared to $|s_1|$. The whole method works as follows: as long as $|s_2|$ is smaller than $2|s_1|$, we check if it is possible to extend any of the long borders of $s_1$. If it is not possible, we replace $s_1$ with the longest prefix of $s$ which ends $s_1[\frac{|s_1|}{2}\twodots |s_1|]$ (we can preprocess such information for all prefixes of $s$ in linear time). When $|s_2|$ exceeds $2|s_1|$, we look for an occurrence of $s_2$ in a prefix of $s$ of length $|s_1|+|s_2|$. All such occurrences create one arithmetic progression due to Lemma~\ref{lemma:few borders}, and it is possible to detect which one is preceded by a suffix of $s_1$ in constant time. More specifically, we show how to implement in constant time the following two primitives. In both cases the method resembles the one from Lemma~\ref{lemma:concatenation occurrence}.

\begin{enumerate}
\item Computing the longest long border of $s_1$ which can be extended with $s_2$ to form a prefix of $s$, if any. First we compute the period $p$ of $s_1$ in constant time due to Lemma~\ref{lemma:longest suffix}, then $p\leq \frac{|s_1|}{2}$ and any long border begins after the $\alpha p$-th letter, for some $\alpha \geq 0$. We compute how far the period extends in both $s$ and $s_2$, this gives us a simple arithmetic condition on the smallest value of $\alpha$. More explicitly, there is either at most one valid $\alpha$, or all are correct.

\item Detecting the rightmost occurrence of $s_2$ in $s$ preceded by a suffix of $s_1$, assuming $|s_2|\geq 2|s_1|$. We begin with finding the first and the second occurrence of $s_2$ in $s$. Assuming we have the corresponding vertex in the suffix tree available, this takes just constant time. We check those (at most) two occurrences naively. There might be many more of them, though. But if the two first occurrences begin before the $|s_1|$-th character, we know that all other interesting occurrences form one arithmetic progression with the known period of $s_2$. We check how far the period extends in $s_1$ (starting from the right end) and $s$ (starting from the first occurrence of $s_2$), this again gives us a simple arithmetic condition on the best possible shift.
\end{enumerate}
\qed
\end{proof}

%Note that the running time from the above lemma stays constant as long as $|s_1|$ is bounded from above by a constant fraction of $|s_2|$.

\section{Constructing balanced grammar}\label{section:constructing}

Recall that a LZ parse is a sequence of triples $(start_i,len_i,next_i)$ for $i=1,2,\ldots,n$. In the not self-referential variant considered in~\cite{CharikarApproximation}, we require that $start_i+len_i-1\leq\sum_{j<i}len_j$ so that each triple refers only to the prefix generated so far. Although such assumption is made by some LZ-based compressors, \cite{Farach} deals with the compressed pattern matching problem in its full generality, allowing self-references. Thus for the sake of completeness we need to construct a balanced grammar from a potentially self-referential LZ parse. It turns out that a small modification of a known method is enough for this task.

\begin{restatable}[see Theorem 1 of~\cite{CharikarApproximation}]{lemma}{lemmabalancedconstruction}
\label{lemma:balanced construction}
Given a (potentially self-referential) LZ parse of size $n$, we can build a $\alpha$-balanced SLP of size $\mathcal{O}(n\log\frac{N}{n})$ describing the same string of length $N$, for any constant $0<\alpha\leq 1-\frac{\sqrt{2}}{2}$. Running time of the construction is proportional to the size of the output.
\end{restatable}

\begin{proof}
At a very high level, the idea of~\cite{CharikarApproximation} is to process the parse from left-to-right. When processing a triple $(start_i,len_i,next_i)$, we already have an $\alpha$-balanced SLP describing the prefix of the whole text corresponding to the previously encountered triples. Because the grammar is balanced, we can define $t[start_i\twodots start_i+len_i-1]$ by introducing a relatively small number of new nonterminals (with small actually meaning small in the amortized sense). Now if we allow the parse to be self-referential, it might happen that $t[start_i\twodots start_i+len_i-1]$ sticks out from the right end of $t[1\twodots \sum_{j=1}^{i-1}len_j]$. In such case we do as follows: let $L=\sum_{j=1}^{i-1}len_j$, and split the fragment corresponding to the current triple into three parts. First we have $t[start_i\twodots L]$, then some repetitions of the same fragment, and then $t[start_i\twodots len_i\bmod (L-start_i+1)]$ followed by a single letter $next_i$. After defining a nonterminal deriving $t[start_i\twodots L]$, we can define a nonterminal deriving the repetitions at the expense of introducing at most $2\log len_i$ new nonterminals. Then we define a nonterminal deriving $t[start_i\twodots len_i\bmod (L-start_i+1)] next_i$. The only change in the analysis of this method is that we might end up adding $\sum_{i=1}^{n} \log len_i$ new nonterminals, which by the concavity of $\log$ is at most $\mathcal{O}(n\log\frac{N}{n})$, and thus does not change the asymptotic upper bound. Note tha~t the authors of~\cite{CharikarApproximation} were not concerned with the computational complexity of their algorithm. Nevertheless, it is easy to see that the only place which cannot be amortized by the number of new nonterminals is finding the corresponding place at the so-called active symbols list and traversing the grammar top-down in order to find the appropriate nonterminal. The former can be implemented by storing the active list in a balanced search tree, adding $\mathcal{O}(n\log n)$ to the time. The latter adds just $\mathcal{O}(n\log N)$ to the whole running time. Hence we can implement the whole method in $\mathcal{O}(n\log N)$. In order to decrease this complexity to just $\mathcal{O}(n\log\frac{N}{n})$, we cut the string into $n$ parts of roughly the same size. Note that this requires that our computational model allows constant time integer division.

Note that the algorithm in~\cite{CharikarApproximation} contains one special case: if the compression ratio is at most $2e$, the trivial grammar is returned. We do the same.
\qed
\end{proof}

As a result we get a context-free grammar in which all nonterminals derive exactly one string, and right sides of all productions are of the form $XY$ with $\frac{\alpha}{1-\alpha}\leq\frac{|X|}{|Y|}\leq\frac{1-\alpha}{\alpha}$. The exact value of $\alpha$ is not important, we only need the fact that both $\frac{|X|}{|Y|}$ and $\frac{|Y|}{|X|}$ are bounded from above. For the sake of concreteness we assume $\alpha=0.25$. We also need to compute $|X|$ for each nonterminal $X$, and to group the nonterminals according to the (rounded down) logarithm of their length, with the base of the logarithm to be chosen later. Note that taking logarithms of large numbers (i.e., substantially longer than $\log n$ bits) is not necessarily a constant time operations in our model. 
%Of course we could preprocess $\log_b x$ for each $x\leq N$, but it introduces an additional $\mathcal{O}(N^\epsilon)$ addend in the running time. 
We can use the fact that the grammar is balanced here: if $X\rightarrow YZ$, then $\log_b|X|\leq\beta+\max\left(\log_b|Y|,\log_b|Z|\right)$ for some constant $\beta$ depending only on $\alpha$ and $b$, and the logarithms can be computed for all nonterminals in a bottom-up fashion using just linear time.

\section{Processing balanced grammar}

While the final goal of this section is a $\mathcal{O}(n\log\frac{N}{n}+m)$ time algorithm, we start with a simple $\mathcal{O}(n\log\frac{N}{n}\log m+m)$ time solution, which then is modified to take just $\mathcal{O}(n\log\frac{N}{n}+m\log m)$, and finally $\mathcal{O}(n\log\frac{N}{n}+m)$ time.

For each nonterminal $X$ we would like to check if the string it represents occurs inside $s$. If it does not, we would like to compute $\prefix(X)$ and $\suffix(X)$, the longest prefix (suffix) which is a suffix (prefix) of the whole $s$. Given such information for all possible nonterminals, we can easily detect an occurrence:

\begin{restatable}{lemma}{lemmafirstoccurrence}
\label{lemma:first occurrence}
If $s$ occurs in a string represented by a SLP then there exists a production $X\rightarrow YZ$ such that $s$ occurs in $\suffix(Y)\prefix(Z)$.
\end{restatable}

\begin{proof}
Consider the leftmost occurrence of $s$. Take the starting symbol $X=S$ and its production $X\rightarrow YZ$. If the leftmost occurrence is completely inside $Y$ or $Z$, repeat with $X$ replaced with $Y$ or $Z$. Otherwise the occurrence crosses the boundary between $Y$ and $Z$, in other words there is a prefix snippet $s[1\twodots i]$ ending $Y$ and a suffix snippet $s[i+1\twodots m]$ starting $Z$. Then $\left|\suffix(Y)\right|\geq i$ and $\left|\prefix(Z)\right|\geq m-i$, and $s$ occurs in $\suffix(Y)\prefix(Z)$.
\qed
\end{proof}

\begin{restatable}{theorem}{theoremslowest}
\label{theorem:slowest}
Given a (potentially self-referential) Lempel-Ziv parse of size $n$ describing a text $t[1\twodots N]$ and a pattern $s[1\twodots m]$, we can detect an occurrence of $s$ inside $t$ deterministically in time $\mathcal{O}(n\log\frac{N}{n}\log m+m)$.
\end{restatable}

\begin{proof}
By Lemma~\ref{lemma:balanced construction} and Lemma~\ref{lemma:first occurrence}, we only have to compute for each nonterminal $X$ its corresponding snippet (if any) and both $\prefix(X)$ and $\suffix(X)$. We process the productions in a bottom-up order. Assume that we have the information concerning $Y$ and $Z$ available and would like to process $X\rightarrow YZ$. If both $Y$ and $Z$ correspond to substrings of $s$, we can apply binary search in the suffix array to check if their concatenation does as well in $\mathcal{O}(\log m)$ steps, each step consisting of two applications of Lemma~\ref{lemma:equality} used to compare the concatenation with a suffix of $s$. To compute $\prefix(X)$ and $\suffix(X)$ in $\mathcal{O}(\log m)$ time we could use Lemma~\ref{lemma:concatenation prefix}. There is one difficulty here, though: we need to know the corresponding node in the suffix tree. To this end we show how to preprocess the tree in linear time so that the corresponding (implicit or explicit) node can be found in $\mathcal{O}(\log m)$ time.

If we allow as much as $\mathcal{O}(m\log m)$ preprocessing time, the implementation is very simple: for each vertex of the suffix tree we construct
a balanced search tree containing all its ancestors sorted according to their depths. Constructing the tree for a vertex requires inserting just
one new element into its parent tree (note that most standard balanced binary search trees can be made persistent so that inserting a new number creates a new copy and does not destroy the old one) and so the whole construction takes $\mathcal{O}(m\log m)$ time. This is too much by a factor
of $\log m$, though. We use the standard micro-macro tree decomposition to remove it. The suffix tree is partitioned into small subtrees by
choosing at most $\frac{m}{\log m}$ macro nodes such that after removing them we get a collection of connected components of at most
logarithmic size. Such partition can be easily found in linear time. Then for each macro node we construct a binary search tree containing
all its macro ancestors sorted according to their depths. There are just $\frac{m}{\log m}$ macro nodes so the whole preprocessing is linear.
To find the ancestor $v$ at depth $d$ we first retrieve the lowest macro ancestor $u$ of $v$ by following at most $\log M$ edges up from $v$. 
If none of the traversed vertices is the answer, we find the macro ancestor of $u$ of largest depth not smaller than $d$ using the binary search tree in
$\mathcal{O}(\log m)$ time. Then retrieving the answer requires following at most $\log m$ edges up from $u$.
\qed
\end{proof}

We would like to remove the $\log m$ factor from the above complexity. It seems that the main difficulty here is that we need to implement a procedure
for detecting if a concatenation of two substrings of $s$ occurs in $s$ as well, and in order to get the claimed running time we would need to answer such queries in constant time after a linear (or close to linear) preprocessing. We overcome this obstacle by choosing to work with an approximation of this information instead and using the fact that the grammar we are working with is balanced.

\begin{definition}
A cover of a nonterminal $X$ is pair of snippets $s[i\twodots i+2^k-1]$ and $s[j\twodots j+2^k-1]$ such that $2^k<|X|\leq 2^{k+1}$, $s[i\twodots i+2^k-1]$ is a prefix of the string represented by $X$, and $s[j\twodots j+2^k-1]$ is a suffix of the string represented by $X$. We call $k$ the order of $X$'s cover.
\end{definition}

We try to find the cover of each nonterminal $X$. If there is none, we know that the string it represents does not occur inside $s$. In such case we compute $\prefix(X)$ and $\suffix(X)$. More precisely, we either:
\begin{enumerate}
\item compute the cover, in such case the string represented by $X$ might or might no occur in $s$,
\item do not compute the cover, in such case the string represented by $X$ does not occur in $s$.
\end{enumerate}
As we will see later, it is possible to extract $\prefix(X)$ and $\suffix(X)$ from the cover of $X$ using Lemma~\ref{lemma:concatenation prefix} in constant time, and the information about $\prefix(X)$ and $\suffix(X)$ for each nonterminal $X$ is enough to detect an occurrence.

To find the covers we process the nonterminals in groups. Nonterminals in the $k$-th group $\mathcal{G}_\ell = \left\{X_1,X_2,\ldots X_s \right\}$ are chosen so that $(\frac{4}{3})^\ell<|X_i|\leq(\frac{4}{3})^{\ell+1}$. The groups are disjoint so $\sum_\ell\left|\mathcal{G}_\ell\right|=\mathcal{O}(n\log\frac{N}{n})$. Furthermore, the partition can be constructed in linear time.
We start with computing the covers of nonterminals in $\mathcal{G}_1$ naively. Then we assume that all nonterminal in $\mathcal{G}_{\ell-1}$ are already processed, and we consider $\mathcal{G}_\ell$. Because the grammar is $0.25$-balanced, if $X_i\rightarrow Y_i Z_i$ then $|Y_i|,|Z_i|\leq\frac{3}{4}|X_i|$, and $Y_i, Z_i$ belong to already processed $\mathcal{G}_{\ell'}$ with $\ell-5 \leq\ell'<\ell$. If for some $Y_i$ or 	$Z_i$ we do not have the corresponding cover, neither must have the corresponding $X_i$, so we use Lemma~\ref{lemma:concatenation prefix} to calculate $\prefix(X_i)$, $\suffix(X_i)$, and remove $X_i$ from $\mathcal{G}_\ell$. For all remaining $X_i$ we are left with the following task: given the covers of $Y_i$ and $Z_i$, compute the cover of $X_i$, or detect that the represented string does not occur in $s$ and so we do not need to compute the cover. Note that the known covers are of order $k$ with $k_{min}=\left\lfloor\ell\log\frac{4}{3}\right\rfloor-3\leq k\leq\left\lceil\ell\log\frac{4}{3}\right\rceil=k_{max}$.

We reduce computing covers to a sequence of batched queries of the form: given a sequence of pairs of snippets $s[i\twodots i+2^{k_1}-1]$, $s[j\twodots j+2^{k_2}-1]$ does their concatenation occur in $s$, and if so, what is the corresponding snippet? We call this merging the pair. For each $\ell$ we will require solving a constant number of such problems with $k_{min}\leq k_1,k_2\leq k_{max}$, each containing $\mathcal{O}(|\mathcal{G}_\ell|)$ queries. We call this problem \proc{Batched-powers-merge}. Before we develop an efficient solution for such question, lets see how it can be used to compute covers.

\begin{restatable}{lemma}{lemmacoversreduction}
\label{lemma:covers reduction}
Computing covers of the nonterminals in any $\mathcal{G}_\ell$ can be reduced in linear time to a constant number of calls to \proc{Batched-powers-merge}, with the number of pairs in each call bounded by $\left|\mathcal{G}_\ell\right|$.
\end{restatable}

\begin{proof}
Recall that for each given pair of snippets we have their covers available, and the orders of those covers are from $\{k,k+1,\ldots,k+4\}$. Consider the situation for a single pair, see Figure~\ref{figure:merging pair}. Let $a,b$ be the cover of the first snippet and $c,d$ the cover of the second snippet. First we merge $b$ and $c$ to get $\text{merge}(b,c)$. Then we extend $a$ to the right and $d$ to the left by merging with the corresponding fragments of $\text{merge}(b,c)$ of length $2^k$, and call the results $\text{extend}(a)$ and $\text{extend}(d)$. Then we would like iteratively extend both $a$ and $d$ with fragments of such length as long as it does not result in sticking out of the considered word $w$. To do that, we need to have the snippets corresponding to those fragments available. Consider the situation for $a$: first we extract the snippets from $\text{merge}(b,c)$, then from $\text{extend}(d)$. We claim that we are always able to perform such extraction: if the next $2^k$ characters fall outside $\text{merge}(b,c)$, the distance to the left boundary of $d$ does not exceed $2^k$ and thus we can use $\text{extend}(d)$. If during this extending procedure the merging fails, the pair does not represent a substring of $s$. Otherwise we get the snippet corresponding to the prefix and suffix of $w$ of lengths $|w|-|w|\bmod 2^k$, which allows us to extract the prefix and suffix of length $2^{k'}$ where $2^{k'}<|w|\leq 2^{k'+1}$, because $k\leq k'$. 

To finish the proof, note that for a single pair we need a constant number of merges. Thus we can do the merging in parallel for all pairs in a constant number of calls to \proc{Batched-powers-merge}.
\qed
\begin{figure}
\centering
\includegraphics[width=\linewidth]{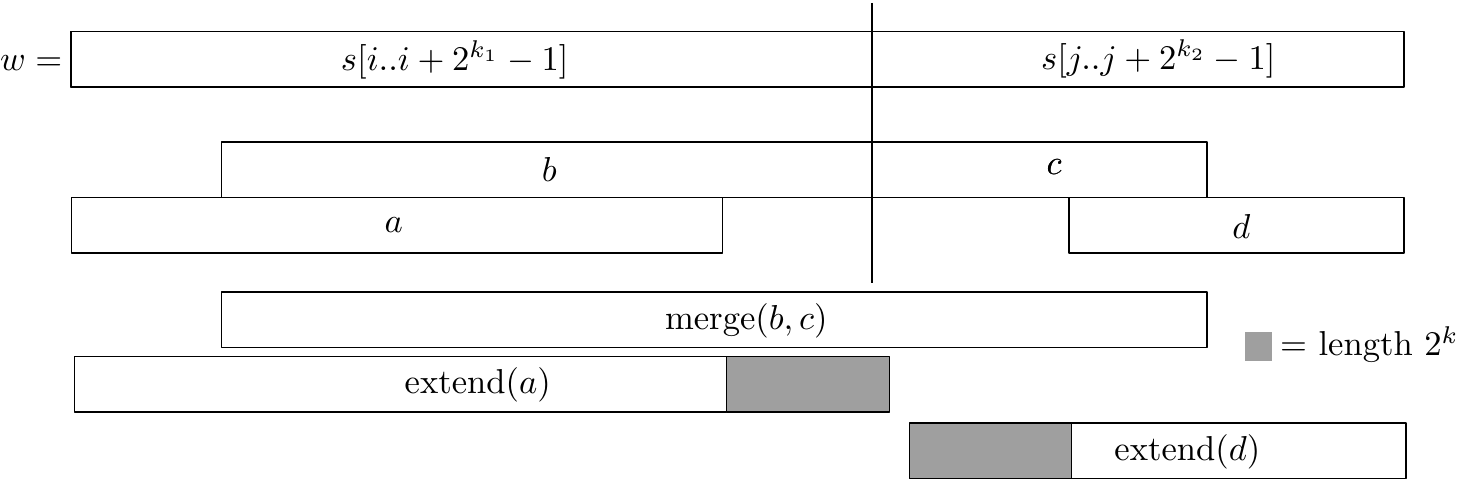}
\caption{Computing cover of a pair of snippets.}
\label{figure:merging pair}
\end{figure}	
\end{proof}

Now we only have to develop the algorithm for \proc{Batched-powers-merge}. A simple solution would be to do a binary search in the suffix array built for $s$ for each pair separately: we can compare $s[i\twodots i+2^{k_1}-1] s[j\twodots j+2^{k_2}-1]$ with any suffix of $s$ in constant time using at most two longest common prefix queries so a single search takes $\mathcal{O}(\log m)$ time, which gets us back to the bounds from Theorem~\ref{theorem:slowest}. In order to get a better running time we aim to exploit the fact that we are given many pairs at once. First observe that we can order all concatenations from a single problem efficiently.

\begin{restatable}{lemma}{lemmapairssort}
\label{lemma:pairs sort}
Given $\mathcal{O}(\left|\mathcal{G}_\ell\right|)$ pairs of words of the form $s[i\twodots i+2^{k_1}-1]$, $s[j\twodots j+2^{k_2}-1]$ with $k_{min}\leq k_1,k_2\leq k_{max}$ we can lexicographically sort their concatenations in time $\mathcal{O}(\left|\mathcal{G}_\ell\right|+m^\epsilon)$ if $|k_{max}-k_{min}|\in\mathcal{O}(1)$.
\end{restatable}

\begin{proof}
We split the words to be sorted into a constant number of chunks of length $2^{k_{min}}$. Then we would like to assign numbers to those chunks so that $\nr(s[i\twodots i+2^{k_{min}}-1])<\nr(s[j\twodots j+2^{k_{min}}-1])$ iff $s[i\twodots i+2^{k_{min}}-1]) <_{lex} s[j\twodots j+2^{k_{min}}-1])$. To compute all $\nr(s[i\twodots i+2^{k_{min}}-1])$ we retrieve the positions of $s[i\twodots m]$ in the suffix array. Then we sort the resulting list of $\mathcal{O}(\left|\mathcal{G}_\ell\right|)$ integers using radix sort, i.e., by $\frac{1}{\epsilon}$ rounds of counting sort. The time required by this sorting is linear plus $\mathcal{O}(m^\epsilon)$. After sorting we scan the list and identify different suffixes with the same prefix of length $2^{k_{min}}$, such suffixes belong to continuous blocks whose boundaries can be identified using longest prefix queries. Then the original task reduces to sorting a list of constant length vectors consisting of integers not exceeding $m$, which can be done efficiently using radix sort.
\qed
\end{proof}

We apply the above lemma to all calls to \proc{Batched-powers-merge} corresponding to nonempty $\mathcal{G}_\ell$. If $(\frac{4}{3})^\ell > m$ then clearly the corresponding $\mathcal{G}_\ell$ is empty, so the total running time of this part is just $\mathcal{O}(m^\epsilon\log m+\sum_\ell\left|\mathcal{G}_\ell\right|)=\mathcal{O}(m+n\log\frac{N}{n})$. Now that the queries in a single call to \proc{Batched-powers-merge} are sorted, instead of performing a separate binary search for each of them we can scan the queries and the suffix array at once, resulting in a $\mathcal{O}(|\mathcal{G}_{\ell}|+m)$ running time for each different $\ell$. This gives us the following total running time.

\begin{theorem}\label{theorem:slower}
Given a (potentially self-referential) Lempel-Ziv parse of size $n$ describing a text $t[1\twodots N]$ and a pattern $s[1\twodots m]$, we can detect an occurrence of $s$ inside $t$ deterministically in time $\mathcal{O}(n\log\frac{N}{n}+m\log m)$.
\end{theorem}

This is still not enough to improve~\cite{Farach} on all possible inputs. We would like to replace $m\log m$ by $m$ in the above complexity by focusing on improving the running time of \proc{Batched-powers-merge}. At a high level the idea is to consider the queries in a single call in sorted order, and for each of them perform a binary search starting from the place where the lexicographically previous pair was found at. This might be still too slow though. To accelerate the search we develop a constant time procedure for locating the fragment of the suffix array corresponding to all occurrences of any $s[i\twodots i+2^k-1]$.

\begin{restatable}{lemma}{lemmafastancestor}
\label{lemma:fast ancestor}
The pattern $s$ can be processed in linear time so that given any $s[i\twodots i+2^k-1]$ we can compute its first and the last occurrence in the suffix array of $s$ in constant time.
\end{restatable}

\begin{proof}
It is enough to show that the suffix tree $T$ built for $s$ can be preprocessed in linear time so that we can locate the (implicit or explicit) vertex corresponding to any fragment which is a power of $2$ in constant time. For that we should locate an ancestor of a given leaf which is at specified depth $2^k$. This can be reduced to the so-called weighted ancestor queries: given a node-weighted tree, with the weights nondecreasing on any root-to-leaf path, preprocess it to find the predecessor of a given weight among the ancestors of $v$ efficiently. Unfortunately, all known solutions for this problem~\cite{FarachHashing,KopelowitzAncestors} give nonconstant query time. We wish to improve this time by abusing the fact that only ancestors at depths $2^k$ are sought. First note that such ancestor is not necessarily an explicit vertex. We start with considering all edges of $T$. For each such edge $e$, we compute the smallest $k$ such that $e$ contains an implicit vertex at depth $2^k$ (there might be none), and split the edge to make it explicit. We call all original vertices at depths being powers of $2$, and all new vertices, marked. For each vertex $v$ we would like to compute the depths of all its marked ancestors, see Figure~\ref{figure:marked}. This can be done in linear time by a single top-bottom transversal, and the information can be stored in a single $\Theta(\log |s|)$-bit word. More precisely, for each vertex $v$ we construct a single word $\text{marked}(v)$ with the $k$-th bit set iff $v$ has a marked ancestor at depth $2^k$. Then we construct $T'=\text{compress}(T)$ containing only the leaves and marked vertices of $T$ by collapsing all maximal fragments of $T$ without such vertices, and build the level ancestor data structure for $T'$~\cite{BenderAncestor} allowing us to find the $k$-th ancestor of any vertex in constant time. Now given $i$ and $k$ we first locate the leaf $v$ corresponding to $s[i\twodots |s|]$ in $T$, then take a look at its bitvector $\text{marked}(v)$. We can compute in constant time $t=\left\{k' > k : k'\in\text{marked}(v)\right\}$ and retrieve the $t$-th ancestor of $v$ in $T'$. Going back to $T$ we get a node with the same (lexicographically) smallest and largest suffix in its subtree as the node corresponding to $s[i\twodots i+2^k-1]$. %In fact can use a simpler solution because the depth of $T'$ is just $\log m$. 

\begin{figure}
\centering
\includegraphics[width=\linewidth]{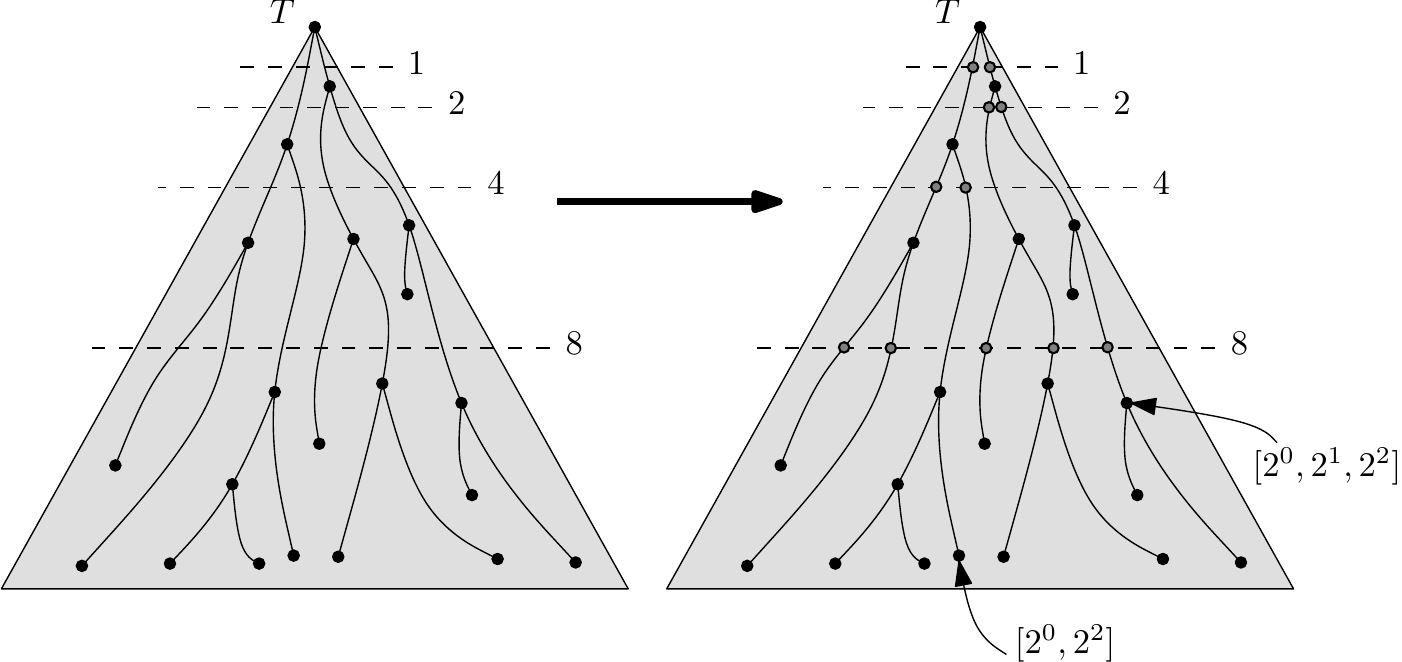}
\caption{Marking vertices at depths $2^k$ in the suffix tree.}
\label{figure:marked}
\end{figure}

While the structure of~\cite{BenderAncestor} does give a constant time answers, we can use a significantly simpler solution building on the fact that the depth of $T'$ is just $\log m$. First we use the standard micro-macro tree decomposition, which gives us a top fragment containing just $\frac{m}{\log m}$ leaves, and a collection of small trees on at most $\log m$ leaves. Note that in this particular case, the total number of vertices cannot be much larger than the number of leaves: the original tree contained vertices with outdegree $1$, then we introduced at most one such vertex at each edge, and then we collapsed some parts of the tree. For each node in the top tree we store all $\log m$ answers explicitly. For each small tree we do as follows: first number its nodes in a depth-first order, then for each node compute a single bitvector containing the numbers of all its ancestors. To find the $k$-th ancestor of a given vertex $v$, we consider two cases.
\begin{enumerate}
\item $v$ belongs to the top tree. Then we have the answer available.
\item $v$ belongs to some small tree. We first check in constant time if its depth in this small tree does not exceed $k$. If it does, we can use the precomputed answers stored for the parent (in the top tree) of the root. Otherwise we take a look at the bitvector corresponding to $v$, and find its $k$-th highest bit set to $1$. Then we retrieve the node corresponding to this depth-first number.
\end{enumerate}
\qed
\end{proof}

Observe that the above lemma can be used to give an optimal solution for a slight relaxation of the {\it substring fingerprints} problem considered in~\cite{FarachHashing}. This problem is defined as follows: given a string $s$, preprocess it to compute any {\it substring hash} $h_s(s[i\twodots j])$ efficiently. We require that:
\begin{enumerate}
\item $h_s(s[i\twodots j])\in[1,\mathcal{O}(|s|^2)]$ so that the values can be operated on efficiently,
\item $h_s(s[i\twodots j])=h_s(s[k\twodots l])$ iff $s[i\twodots j]=s[k\twodots l]$.
\end{enumerate}
If we allow the range of $h_s$ to be slightly larger, say $\mathcal{O}(|s|^3)$, a direct application of the above lemma allows us to evaluate the fingerprints in constant time.% after a linear preprocessing.
%While such range is not optimal, it allows constant time evaluation which should be enough to replace the results of~\cite{FarachHashing}.

\begin{restatable}{theorem}{theoremfingerprints}
\label{theorem:fingerprints}
Substring fingerprints of size $\mathcal{O}(|s|^3)$ can be computed in constant time after a linear time preprocessing.
\end{restatable}

\begin{proof}
First we apply the preprocessing from Lemma~\ref{lemma:fast ancestor} to $s$. We also store $\left\lfloor\log x\right\rfloor$ for any $1\leq x\leq |s|$. Then given a query $s[i\twodots j]$ we compute $k=\left\lfloor\log(j-i+1) \right\rfloor$ and using constant time level ancestors queries we locate the lowest existing ancestors of both $s[i\twodots i+2^k-1]$ and $s[j-2^k+1\twodots j]$ in the suffix tree. Then $h_s(s[i\twodots j])$ is a triple containing $j-i+1$ and those two ancestors.
\qed
\end{proof}

Now getting back to the original question, the input to \proc{Batched-power-merge} is a sequence of pairs of snippets $w_{1},w_{2},\ldots,w_{\left|\mathcal{G}_\ell\right|}$. By Lemma~\ref{lemma:pairs sort} we can consider them in a sorted order. For each such pair $w=s[i\twodots i+2^{k_1}-1] s[j\twodots j+2^{k_2}-1]$, we first look up the fragment of the suffix array corresponding to its prefix $s[i\twodots i+2^{k_{min}}-1]$ using Lemma~\ref{lemma:fast ancestor}. Then we apply binary search in this fragment, with the exception that if the previous binary search was in this fragment as well, we start from the position it finished, not the beginning of the fragment. Additionally, the binary search is performed from the beginning and the end of the interval at the same time, see \proc{Two-way-binary-search}. If the initial interval is $[a,b]$ and the position we are after is $r$, such modified search uses just $\mathcal{O}(\log\min(r-a+1,b-r+1))$ applications of Lemma~\ref{lemma:equality} instead of $\mathcal{O}(\log(b-a+1))$ time, which is important.

\begin{algorithm}
\caption{\proc{Two-way-binary-search}$(a,b,w)$}
\begin{algorithmic}[1]
\State $x \gets a$, $y \gets b$
\State $k \gets 1$
\While{$2^k \leq b-a$}
  \If{$w <_{lex} s[SA[a + 2^k]]$}
    \State $y \gets a + 2^k$
    \State \textbf{break}
  \EndIf
  \If{$s[SA[b - 2^k]] <_{lex} w$}
    \State $x \gets b - 2^k$
    \State \textbf{break}
  \EndIf
  \State $k \gets k + 1$
\EndWhile
\State $r \gets \text{binary search for } w \text{ in } s[SA[x]\twodots |s|],s[SA[x+1]\twodots |s|],\ldots,s[SA[y]\twodots |s|]$
\State \Return $r$
\end{algorithmic}
\end{algorithm}

While a single binary search might require a non-constant time, we will show that their amortized complexity is constant. To analyze the whole sequence of those searches, we keep a partition of the whole $[1,|s|]$ into a number of disjoint intervals. Doing a single search splits at most one interval into two parts at the position of the first occurrence. If the first occurrence is exactly at an already existing boundary, there is no split, otherwise we say that those two smaller intervals have been created in phase $k_{min}$ (recall that $k_{min}$ linearly depends on $\ell$), and intervals created in phase $k_{min}$ are kept in a list $I_{k_{min}}$. We do not want to split an interval more than once and hence each call to \proc{Batched-powers-merge} starts with finding for each $w_{i}$ its corresponding interval in $I_{k_{min}}$. After processing all concatenations, we add the new intervals to $I_{k_{min}}$ and prune it to contain the intervals which are minimal under inclusion. Scanning and pruning $I_{k_{min}}$ takes linear time in its size, and we show that this size is small.
%(and that the amortized complexity of a single binary search is constant).
\begin{algorithm}
\caption{\proc{Batched-powers-merge}$(w_1,w_2,\ldots,w_{\left|\mathcal{G}_\ell\right|})$}
\begin{algorithmic}[1]
\State sort all $w_i$ \label{line:sort input} \Comment{{\bf Lemma~\ref{lemma:pairs sort}}}
\State scan $I_{k_{min}}$ to find the intervals containing $w_i$ \label{line:scan input}
\State $L \gets \emptyset$
\State $r_0 \gets 1$
\For{$i \gets 1$ to $\left|\mathcal{G}_\ell\right|$}
  \State $[a,b] \gets \text{the interval corresponding to } w_i[1\twodots 2^{k_{min}}] \text{ in } SA$ \Comment{{\bf Lemma~\ref{lemma:fast ancestor}}}
  \State choose $[c,d]\in I_{k_{min}}$ containing the first occurrence of $w_i$ in SA \label{line:choose first}
  \If{$[c,d]$ is defined}
    \State $a \gets \max(a,c)$
    \State $b \gets \min(b,d)$
  \EndIf
  \State $a \gets \max(r_{i-1}, a)$
  \State $r_i \gets \proc{Two-way-binary-search}(a,b,w_i)$ \label{line:two way search}
  \State add $[a,r_i]$ and $[r_i,b]$ to $L$
\EndFor
\State sort $L$ and merge it with $I_{k_{min}}$, removing non-minimal intervals \label{line:sorting list}
\State \Return all answers $r_i$
\end{algorithmic}
\end{algorithm}
\begin{lemma}\label{lemma:batched merges}
All $\mathcal{O}(\log m)$ calls to \proc{Batched-powers-merge} run in total time $\mathcal{O}(m+\sum_\ell\left|\mathcal{G}_\ell\right|)$.
\end{lemma}

\begin{proof}
First note that the sorting in line~\ref{line:sorting list} can be performed in time $\mathcal{O}(m^\epsilon+\left|I_{k_{min}} \right|+\left|\mathcal{G}_\ell\right|)$ using radix sort. Line~\ref{line:sort input} takes time $\mathcal{O}(m^\epsilon+\left|\mathcal{G}_\ell\right|)$ due to Lemma~\ref{lemma:pairs sort}, and line~\ref{line:scan input} requires $\mathcal{O}(\left|I_{k_{min}}\right|+\left|\mathcal{G}_\ell\right|)$. All executions of line~\ref{line:choose first} take time $\mathcal{O}(\left|I_{k_{min}}\right|)$ because the words $w_i$ are already sorted. For the time being assume that the binary search in line~\ref{line:two way search} is for free. Then the total complexity becomes $\mathcal{O}(\sum_i m^\epsilon + \left|I^{(i)}_{k_{min}} \right|+\left|\mathcal{G}_\ell\right|)$ where $\left|I^{(i)}_{k_{min}}\right|$ is the size of $I_{k_{min}}$ just before the $i$-th call to \proc{Batched-powers-merge}. There is a constant number of those calls for each value of $1\leq\ell\leq m$, and each $k_{min}$ corresponds to at most constant number of different continuous values of $\ell$, thus the sum is in fact $\mathcal{O}(m+\sum_\ell\left|\mathcal{G}_\ell\right|)$.

To finish the proof we have to bound the time taken by all binary searches. For that to happen we will view the intervals as vertices of a tree. Whenever performing a binary search splits an interval into two, we add a left and right child to the corresponding leaf $v$, see Figure~\ref{figure:search}. The {\it rank} $\rank(v)$ of a vertex $v$ is the rounded logarithm of its {\it weight}, which is the length of the corresponding interval. Then the cost of line~\ref{line:two way search} is simply $\mathcal{O}(1+\min(\rank(\text{left}(v)),\rank(\text{right}(v))))$ where $\text{left}(v)$ and $\text{right}(v)$ are the left and right child of $v$, respectively. Hence we should bound the sum $\sum_{v} \min(\rank(\text{left}(v)),\rank(\text{right}(v)))$, where $v$ is a non-leaf. We say that a vertex is {\it charged} when its weight does not exceed the weight of its brother. Now we claim that there are at most $\frac{m}{2^k}$ charged vertices of rank $k$: assume that there are $u$ and $v$ such that $u$ is an ancestor of $v$, both are charged and of rank $k$, then weight of $v$ plus weight of its brother is at least twice as large as the weight of $v$ alone, thus the rank of their parent is larger than the rank of $v$, contradiction. So all charged vertices of the same rank correspond to disjoint intervals, and there cannot be more than $\frac{m}{2^k}$ disjoint intervals of length at least $2^k$ on a segment of length $m$. Bounding the sum gives the claim:
\vspace{-0.3cm}
$$\hspace{0.5cm}\sum_{v} \min(\rank(\text{left}(v)),\rank(\text{right}(v))) \leq \sum_{k\geq 0}^{\log m} k\frac{m}{2^k} \leq m\sum_{k\geq 0}^{\infty}\frac{k}{2^k}=2m\hspace{1cm}\qed$$
\vspace{-0.7cm}
\begin{figure}
\centering
\includegraphics[width=\linewidth]{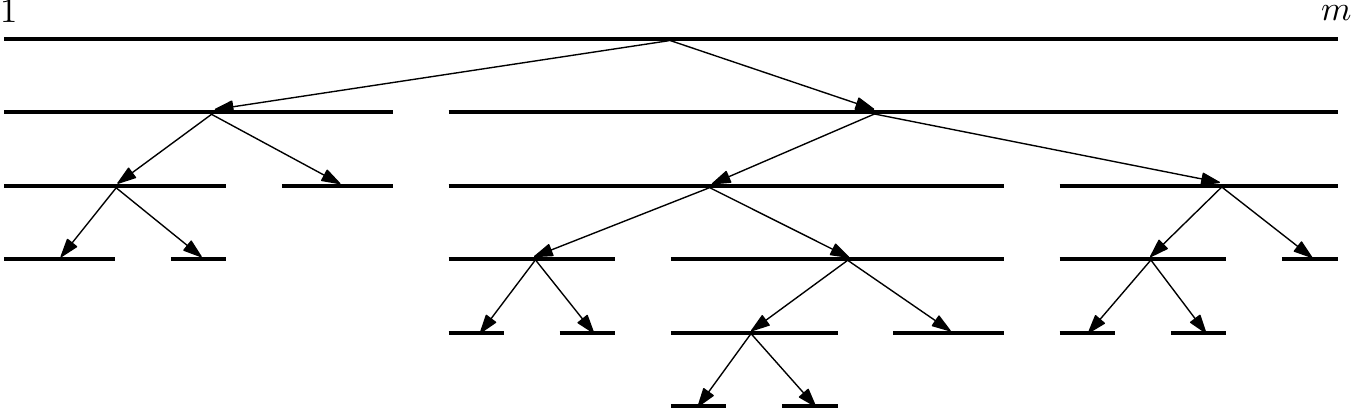}
\caption{Interpreting the intervals as a tree.}
\vspace{-0.6cm}
\label{figure:search}
\end{figure}
\end{proof}

Hence for all productions $X\rightarrow YZ$ such that we have the cover of both $Y$ and $Z$, we either computed the cover of $X$ or decided that there is none. If for a production we cannot find the cover of $X$, we compute $\prefix(X), \suffix(X)$ given the covers of $Y$ and $Z$
using a few applications of  Lemma~\ref{lemma:concatenation prefix} with carefully chosen arguments.

\begin{restatable}{lemma}{lemmaprefixfromcovers}
\label{lemma:prefix from covers}
Given the covers of $Y$ and $Z$, we can compute $\prefix(X)$ and $\suffix(X)$ in constant time as long as $\frac{|Y|}{|Z|}$ and $\frac{|Z|}{|Y|}$ are bounded from above by a constant. To compute $\prefix(X)$ we can use $\prefix(Z)$ instead of the cover of $Z$, and $\suffix(X)$ can be replaced with $\suffix(Y)$ instead of the cover of $Y$.
\end{restatable}

\begin{proof}
It is enough to consider $\prefix(X)$. The idea is to use a few application of Lemma~\ref{lemma:concatenation prefix} with carefully chosen arguments, see Figure~\ref{figure:prefix}. More specifically, let $a,b$ and $c,d$ be the covers of $Y$ and $Z$, respectively. First we locate the vertex corresponding to $d$ in the suffix tree, due to Lemma~\ref{lemma:fast ancestor} and $|d|=2^k$ it takes constant time, then:
\begin{enumerate}
\item[(1)] apply Lemma~\ref{lemma:longest suffix} to compute $\prefix(d)$ if we have the cover of $Z$, otherwise take the known $\prefix(Z)$ and go to (3),
\item[(2)] apply Lemma~\ref{lemma:concatenation prefix} to $c$ and $\prefix(d)$ without the first $|c|+|d|-|Z|$ letters to get $\prefix_1$,
\item[(3)] apply Lemma~\ref{lemma:concatenation prefix} to $b$ and $\prefix_1$ to get $\prefix_2$,
\item[(4)] apply Lemma~\ref{lemma:concatenation prefix} to $a$ and $\prefix_2$ without the first $|a|+|b|-|Y|$ letters to get the desired answer $\prefix_3$.
\end{enumerate}
Note that whenever we apply the lemma to two words $u$ and $v$, $|v|$ is a power of $2$ and so we can use Lemma~\ref{lemma:fast ancestor} to locate its corresponding node in constant time. Also, it holds that $|u|\geq\frac{\min(|Y|,|Z|)}{2}$ and $|v|\leq |Y|+|Z|$ and so the running time is bounded by:
$$
 \max\left(1,\log\frac{|v|}{|u|}\right) \leq \max\left(1, \log\left(\frac{|Y|+|Z|}{\min(|Y|,|Z|)}\right)\right) = \log\left(1+\frac{\max(|Y|,|Z|)}{\min(|Y|,|Z|)}\right)
$$
which is $\mathcal{O}(1)$.
\qed

\begin{figure}
\centering
\includegraphics[width=\linewidth]{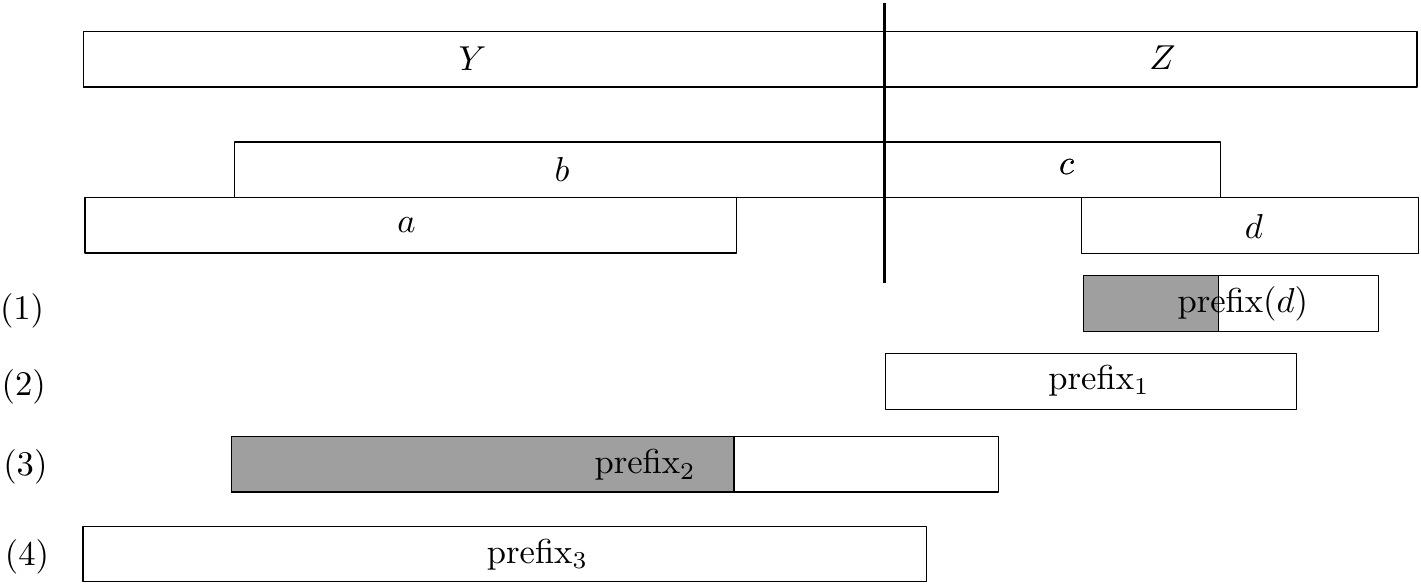}
\caption{Computing $\prefix(X)$ given the covers of $Y$ and $Z$.}
\label{figure:prefix}
\end{figure}

\end{proof}

\begin{theorem}\label{theorem:balanced occurrence}
Given a 0.25-balanced SLP of size $\mathcal{O}(n\log\frac{N}{n})$ and a pattern $s[1\twodots m]$, we can detect an occurrence of $s$ in the represented text in time $\mathcal{O}(n\log\frac{N}{n}+m)$.
\end{theorem}

\begin{proof}
By Lemma~\ref{lemma:covers reduction} and Lemma~\ref{lemma:batched merges} we compute the covers of all nonterminals which represent subwords of $s$ in time $\mathcal{O}(n\log\frac{N}{n}+m)$. For the remaining nonterminals $X$ we use Lemma~\ref{lemma:prefix from covers} to compute $\prefix(X)$ and $\suffix(X)$ in total linear time considering the nonterminals in bottom-up order. Then due to Lemma~\ref{lemma:first occurrence} if there is an occurrence of $s$, there is an occurrence in $\prefix(Y)\suffix(Z)$ for some production $X\rightarrow YZ$. We consider every nonterminal $X$, either lookup the already computed $\prefix(Y)$ and $\suffix(Z)$ or compute them using the known covers and Lemma~\ref{lemma:prefix from covers}, and use Lemma~\ref{lemma:concatenation occurrence} to detect a possible occurrence.
\qed
\end{proof}

%By using Lemma~\ref{lemma:balanced construction} and Theorem~\ref{theorem:balanced occurrence} we get the final result.

\begin{theorem}
Given a (potentially self-referential) Lempel-Ziv parse of size $n$ describing a text $t[1\twodots N]$ and a pattern $s[1\twodots m]$, we can detect an occurrence of $s$ inside $t$ deterministically in time $\mathcal{O}(n\log\frac{N}{n}+m)$.
\end{theorem}

\section{Conclusions}

Recall that in order to guarantee a $\mathcal{O}(n\log\frac{N}{n}+m)$ running time, it was necessary to use integer division in the proof of Lemma~\ref{lemma:balanced construction}. This was the only such place, though. If we assume that integer division is not allowed, and the only
operations on the integers $start_i,len_i$ appearing in the input triples are addition, subtraction, multiplication and comparing with $0$
(which are the only operations used by the $\mathcal{O}(n\log N+m)$ version of our algorithm), we can prove a matching lower bound by looking at
the corresponding algebraic computation trees. More precisely, using standard tools~\cite{Lubiw} one can show that the depth of such tree which recognizes the set of integers $t,x_1,x_2,\ldots,x_n$ such that for all $i$ it holds that $x_i=(2\alpha_i+1)t+\beta_i$ with $0\leq\beta_i<t$ 
and $0\leq \alpha_i< N$ is $\Omega(n\log N)$. On the other hand, one can construct a self-referential LZ of constant size deriving
$(1^t0^t)^N$. Hence one can also construct a LZ of size $\mathcal{O}(n)$ deriving $(1^t0^t)^N b_1 1 \ldots b_n 1$ where
$b_i= \left\lfloor\frac{x_i}{t}\right\rfloor\bmod 2$. This string does not contain $11$ as a substring iff all $x_i$ are of the form
$x_i=(2\alpha_i+1)t+\beta_i$ and the lower bound follows.

%There is no hope to improve this result unless we either improve the balanced construction (which probably would result in getting a better smallest %grammar approximation method, and it is believed that a substantial improvement here is rather unlikely), or use a completely different approach.

%We state an interesting open problem: can the range of substring fingerprints from Theorem~\ref{theorem:fingerprints} be improved to %$\mathcal{O}(|s|^2)$ using a similar idea? We believe that it might be possible by a deeper insight into the the periodicity of given fragments of $s$, %and using a different hashing scheme if the periodicity of a given fragment is high.

\bibliographystyle{abbrv}
\bibliography{biblio}

\end{document}